\documentclass[draftcls,onecolumn]{IEEEtran}

\IEEEoverridecommandlockouts                              % This command is only needed if
\usepackage{graphics} % for pdf, bitmapped graphics files
\usepackage{epsfig} % for postscript graphics files
\usepackage{mathptmx} % assumes new font selection scheme installed
\usepackage{times} % assumes new font selection scheme installed
\usepackage{amsmath} % assumes amsmath package installed
\usepackage{amssymb}  % assumes amsmath package installed
\usepackage{multicol}
\usepackage{multirow}
\usepackage{color}
\usepackage{mathrsfs}
\usepackage{xparse}
\usepackage{bm}
\usepackage{cite}
\usepackage{verbatim}
\usepackage{pifont}
\usepackage{caption}
\usepackage[linesnumbered,ruled,vlined]{algorithm2e}
\usepackage[papersize={8.5in,11in}, left=0.75in, right=0.75in, top=1in, bottom=0.82in]{geometry}
\usepackage{lipsum}
\usepackage{amsthm}
\usepackage{setspace}
\newtheorem{theorem}{Theorem}
\newtheorem*{theorem*}{Theorem}
\newtheorem{lemma}{Lemma}

\newtheorem*{lemma*}{Lemma}
\newtheorem*{remark*}{Remark}
\newtheorem*{fact*}{Fact}
\newtheorem*{proof*}{Proof}
\newcommand*{\QEDA}{\hfill\ensuremath{\blacksquare}}%

\newtheorem{proposition}{Proposition}

\DeclareMathOperator{\E}{\mathbb{E}}
\DeclareMathOperator{\Var}{\text{Var}}

\newlength\myindent % define a new length \myindent
\begin{document}
	%\begin{titlepage}
	\title{An Optimal Treatment Assignment Strategy to Evaluate Demand Response Effect}

	\author{Pan~Li,~\IEEEmembership{Student~Member,~IEEE,}~and~Baosen Zhang,~\IEEEmembership{Member,~IEEE}
		\thanks{This work was supported by NSF grant CNS-1544160 and the University of Washington Clean Energy Institute.}
		\thanks{The authors are with the Department of Electrical Engineering, University of Washington, Seattle, WA 98195, USA (e-mail: \{pli69,zhangbao\}@uw.edu).}
		\thanks{Partial results appeared in an earlier version of this paper presented in the Allerton Conference, 2016.}
	}

	\markboth{IEEE Transactions on Smart Grid}%
	{Li\MakeLowercase{\textit{et al.}}: TBD}

	\maketitle
	\thispagestyle{empty}
	\pagestyle{empty}
	%%%%%%%%%%%%%%%%%%%%%%%%%%%%%%%%%%%%%%%%%%%%%%%%%%%%%%%%%%%%%%%%%%%%%%%%%%%%%%%%
	\begin{abstract}
		Demand response is designed to motivate electricity customers to modify their loads at critical time periods. The accurate estimation of impact of demand response signals to customers' consumption is central to any successful program. In practice, learning these response is nontrivial because operators can only send a limited number of signals. In addition, customer behavior also depends on a large number of exogenous covariates. These two features lead to a high dimensional inference problem with limited number of observations. In this paper, we formulate this problem by using a multivariate linear model and adopt an experimental design approach to estimate the impact of demand response signals. We show that randomized assignment, which is widely used to estimate the average treatment effect, is not efficient in reducing the variance of the estimator when a large number of covariates is present. In contrast, we present a tractable algorithm that strategically assigns demand response signals to customers. This algorithm achieves the optimal reduction in estimation variance, independent of the number of covariates. The results are validated from simulations on synthetic data as well as simulated building data.

	\end{abstract}
	%%%%%%%%%%%%%%%%%%%%%%%%%%%%%%%%%%%%%%%%%%%%%%%%%%%%%%%%%%%%%%%%%%%%%%%%%%%%%%%%

	%\newgeometry{top=0.75in,bottom=0.75in,right=0.75in,left=0.75in}
	\section{INTRODUCTION}
	As more uncertain and intermittent renewable resources are integrated into the power system, operators are increasingly exploring flexibility in customers' consumptions to balance supply and demand. This operation is commonly known as \emph{demand response (DR)}. In a typical implementation of DR programs, customers receive a DR signal to elicit a change in their consumptions. This signal can be a modification of electricity prices or simply a message requesting a change in consumption~\cite{Siano2014}. An effective DR program improves the efficiency and sustainability of power systems and is a central pillar of the envisioned smartgrid~\cite{DR0,DR2,Dai,LiEtAl2011}.

	A natural question about demand response is quantifying the impact of a DR signal. That is, if a DR signal is sent to a subset of the users, what is the change in these users' consumptions because of that signal? An accurate estimate of this change is central to the operation of demand response programs: if not enough change in demand is elicited, other measures need to be taken; if too much change is elicited, the program is inefficient.

	Most of existing work in this area of demand response have approached the problem from a market optimization point of view. For example, authors in~\cite{LiEtAl2011,QianEtAl2013} considered how to optimize the social welfare; and authors in~\cite{SaadEtAl2012} have considered how to create an efficient market for demand response. In all these settings, customers' responses are captured by well-defined utility functions. These functions are assumed to be known to the operators, or at least to the users themselves.

	In practice, these market based approaches can be difficult to implement because customers often do not have a clear model of their own utilities. For example, consider a household with a smart energy management system (e.g., a NEST thermostat). This household will respond to a DR signal, but the response can be a complicated function of the current conditions in the household--e.g., temperature, appliances that are on, number of people at home and so on-- and the user may not be consciously aware of the households' utility function. Therefore the operator needs to \emph{learn customers' responses} from past history. Furthermore, because of the advancement of household sensors, this response need to be learned under a possibly \emph{high dimensional setting}.

By performing enough experiments with enough customers, that is, sending enough DR signals, the operator will eventually learn the users' response with accuracy. Repeated experimentation with large group of customers, however, is impractical for two reasons. The first is that operators only sends out DR signals if the demand need to be modified appreciably, and this event does not occur all that often in the power system. The second is that because of use fatigue~\cite{HolyheadEtAl2015}, most utilities have agreements with their customers that a single household will only receive a limited number of DR requests~\cite{BalijepalliEtAl2011}. Under these \emph{limited data} regimes, accurately learning the response of users is nontrivial.
% In fact, we show that under well-established strategies such as randomly selecting users to receive DR signals, the response of the users cannot b

	In this paper, we adopt an \emph{experimental design approach} to the problem estimating the response of users to DR signals.   We design user selection algorithms where by carefully choosing the users that receive DR signals, the maximum information about the system response can be learned.

	In this paper we consider the linear setting, where a user's consumption is a linear function of a set of variables and the DR signal. We refer to the former abstractly as a user's \emph{covariates}, where they could represent measurements such as temperature, appliance status, building type, behavior patterns and etc. Because of the explosive growth in sensing devices, we are particularly interested in the high dimensional case, that is, where users have a large number of covariates. We assume the impact of DR signals is additive to the original consumption behavior. Using the language of experimental design, we regard DR signal as a \emph{treatment} and a user receiving a DR signal as being \emph{assigned this treatment}.  Then learning the DR response of the users is equivalent to learning the \emph{average treatment effect}, which is the average response of the customers to the treatment~\cite{Holland1986}. The metric with respect to the estimation of the treatment effect is the \emph{variance} of the estimator as more experiments are performed.
	 %
	%  assume that consumption behavior can be captured as a linear function of users' covariates and we design a DR assignment strategy such that the unknown impact of DR can be estimated accurately. These covariates represent other factors that could impact electricity consumption, from temperature to the behavior patterns of users.
	 %
	%  We assume an additive linear model,   The main quantity of interest of this paper is the \emph{variance} of the estimator of the average treatment effect as more and more experiments are performed.
	%
	Randomized trial is usually thought as the ``gold standard'' in these types of models mainly due to the fact that randomly assigning treatments to users removes the effect of confounding factors and provides a consistent estimate of the treatment effect. In the presence of many covariates, however, random assignment can be extremely inefficient. In fact, as we show in this paper, in the high dimensional setting, random assignment \emph{does not reduce the variance} of the estimate of the average treatment effect, even as the number of treatment grows without bound. Instead, following the outline in \cite{ABtest}, we design a selection scheme users are picked based on their covariates to be treated.

	Suppose there are $n$ total users in the system. Under the linear models considered in this paper, the best possible lower bound on the rate of variance reduction is $\Omega(1/n)$,\footnote{no estimator can reduce variance faster than $1/n$} given by considering the Fisher information~\cite{Fisher1935}. As discussed above, under high dimensional settings, a randomized algorithm can only achieve $\Omega(1)$, even when are large number of fraction of users are assigned DR signals. The main contributions of this paper are:
	\begin{enumerate}
		\item We show if the number of users selected is a constant fraction of the total number of users, there exists user assignments that achieve a variance reduction rate of $\Theta(1/n)$. This rate is independent of the dimension of the covariates, as long as it is less than $n$.
		\item We develop a tractable user assignment algorithm. This algorithm is obtained by converting the variance reduction problem in a densest-cut problem on a graph~\cite{KortsarzPeleg8,FeigLangber6,SDP,YeZhang19}.
	\end{enumerate}

Our approach differs from previous effort in learning demand response in one important regard. In previous studies, the focus was on training the best predictive model and subtracting out the predicted consumption from the measured consumption~\cite{Zhou16,Brodersen15}. In our approach, we do not ever learn a predictive model, in the sense that we do not learn the relationship between the covariates and the consumption. Rather, focus on learning a single parameter: the response to the DR signal.

	The results in \cite{ABtest} act as an impetus to this paper. The main difference is that in \cite{ABtest}, the users are assigned a treatment of $\pm 1$, therefore some information is always conveyed by this assignment. In our model, the users are assigned either $1$ (receives DR signal) or $0$ (no signal). Therefore, for the users assigned $0$, we do not obtain any information about the impact of DR signals. This makes the problem much more technically challenging, and consequently we only consider the offline assignment problem whereas \cite{ABtest} also considers the online assignment problem. There are extensive literature on average treatment effect estimation, and the interested reader can refer to \cite{causalreview1,causalreview2} and references within.

	The rest of the paper is organized as follows. Section \ref{problem} introduces the preliminaries and the problem throughout this paper. Section \ref{random} presents the variance of the estimator obtained by random assignment. Section \ref{optimal} further presents the variance of the estimator by optimal assignment, followed by a tractable algorithm presented in Section \ref{sec:SDP}. Section \ref{simulation} details the simulation results obtained by either random assignment and optimal assignment. Section \ref{conclusion} concludes the paper.

	\section{PRELIMINARIES AND PROBLEM FORMULATION}\label{problem}

	In this paper we assume that a user's consumption is given by a linear model. Let $x_i \in \{0,1\}$ denotes a binary DR signal, where $1$ represents that a signal is sent to user $i$ and $0$ presents that no signal is sent. A covariate vector $\bm z_i$ is also associated with a user, representing available side information. For example, side information may include local temperature, user's household size, and number of electrical vehicles and so on. We denote the dimension of the covariate vector by $p$, and assume the last component is $1$, which is the intercept. Let $y_i$ denote the consumption of user $i$, which is given as
	\begin{equation}\label{prob:problem}
	y_i  = \beta x_i + \bm{\gamma}^T \bm{z}_i +\epsilon_i,
	\end{equation}
	where $\epsilon_i$ is white noise with variance $\sigma^2 = 1$ (for convenience). The coefficient $\beta$ is the impact of the DR signal and estimating it efficiently is the goal of the paper. The coefficient $\bm \gamma$ represents the effect of the covariate vectors. The main technical challenge is to accurately estimate the coefficient of interest $\beta$, even when $\gamma$ is high dimensional. For analytical simplicity, we assume that the entries of $\bm z_i$ are drawn as i.i.d. Gaussian random variables~(possibly after centering and rescaling). In simulations (Section \ref{simulation}), we show that the results holds for other types of distributions as well.

	We assume there are $n$ total users. In this model, a single user that receives two demand response signals at two different times is equivalent to two users each receiving a demand response signal. Therefore, we suppress the time dimension and label all users by $i$. Note that in \eqref{prob:problem}, all users share a common response $\beta$ to DR signals.

	% We can assume that these covariates follow i.i.d. Gaussian distribution after proper scaling.

	% In this paper we estimate the DR effect (treatment effect) by forming a linear regression of consumption data on user covariates $\bm{Z}_i$ and a binary DR signal $x_i \in \{0,1\}$, where index $i$ denotes user $i$ in the dataset. This model is presented in \eqref{prob:problem},where $\epsilon_i$ is white noise with variance $\sigma^2$. The DR signal takes on value 1 if user $i$ is included into the DR program and 0 if not. The covariates $\bm{Z}$ has dimension $p$, including one column of ones as intercept. Without loss of generality, let us assume that the last column of $\bm{Z}$ is the intercept. The user covariates include all necessary side information relevant to predict consumption. For example, side information may include local temperature, user's household size and number of Electric Vehicles (EV). We can assume that these covariates follow i.i.d. Gaussian distribution after proper scaling. In addition,  the DR signal has a constant effect across the population and this effect is denoted by $\beta$. The model is presented in \eqref{prob:problem},
	%
	%
	% \begin{equation}\label{prob:problem}
	% y_i  = x_i\beta +\bm{Z}_i\bm{\gamma}+\epsilon_i.
	% \end{equation}

	We denote the estimate of $\beta$ by $\hat{\beta}$. The value of $\hat{\beta}$ is a function of the DR assignments, that is, the value of the $x_i$'s.  Under the linear setting in \eqref{prob:problem}, the ordinary least square (ols) estimator $\hat{\beta}$ of $\beta$  is unbiased for all possible allocations of DR assignments, $\hat{\beta}$ is centered at the true value $\beta$. The natural measure of performance is then the variance: $\Var \hat{\beta}$. With some simple linear algebra, the variance of $\hat{\beta}$ is given by~\cite{LinReg2012}:
	\begin{equation}\label{eq:varbeta}
	\Var \hat{\beta} = \frac{\sigma^2}{\bm{x}^{\text{T}}P_{\bm{Z}^{\perp}}\bm{x}},
	\end{equation}
	where $P_{\bm{Z}^{\perp}} = I-\bm{Z}(\bm{Z}^{\text{T}}\bm{Z})^{-1}\bm{Z}^{\text{T}}$. The $i$'th row of the data matrix $\bm{Z}$ is given by $\bm{z}_i^T$.
	We adopt the notation that $\bm{Z}_{n,p}$ denotes a matrix $\bm{Z}$ that has $n$ rows and $p$ columns, while $\bm{Z}_{i:j}$ denotes the $i^{th}$ to $j^{th}$ column of a matrix $\bm{Z}$, where $i \leq j$.

	We are primarily interested in the setting where an operator can assign a limited number of $x_i$'s to be $1$. This setting reflects the limit in budget of an operator in sending DR signals. Specifically, let $k$ be the total number of DR signal that can be sent. The goal of the operator is to strategically assign $k$ $x_i$'s to be $1$ such that the variance of $\hat{\beta}$ is minimized. In particular, we are interested in the rate of reduction of $\hat{\beta}$ as $n$ increases and in settings where $k/n$ is a constant.

	% Besides randomly picking users into the treatment group and control group, we are interested in an optimal assignment which minimizes the quantity in \eqref{eq:varbeta}, given that $k$ out of $n$ subjects are chosen into the treatment group. As in the context of DR assignment, this is mainly due to the limit of the budget from the system operator and the best $k$ users are chosen to receive the DR signal. This constraint can be written as:
	% \begin{equation}
	% 	\sum x_i = k, \quad\forall i,
	% \end{equation}
	% where $x_i$ is the $i^{th}$ element in the vector $\bm{x}$.
	%
	% Since least square estimators are unbiased estimators for a linear regression model ($\hat{\beta}$ is unbiased given any possible assignment), we are interested in whether the variance of the estimator $\hat{\beta}$ is indeed improved with the proposed assignment strategy, and if yes, by how much $\frac{\sigma^2}{\bm{x}^{\text{T}}P_{\bm{Z}^{\perp}}\bm{x}}$ is reduced.

	From \eqref{eq:varbeta}, minimizing the variance of $\hat{\beta}$ is equivalent to maximizing  the quantity $\bm{x}^{\text{T}}P_{\bm{Z}^{\perp}}\bm{x}$, and we focus on the latter quantity in the rest of the paper due to notational convenience. Two types of algorithms are of interest: i) the standard random assignment where each $x_i$ is chosen to be $1$ or $0$ with probability $k/n$, and ii) an optimal assignment procedure where $x_i$'s are chosen to maximize $\bm{x}^{\text{T}}P_{\bm{Z}^{\perp}}\bm{x}$. Both algorithms face the constrain that only $k$ out of $n$ $x_i$'s can be assigned to be $1$.  We characterize \emph{growth rate} of quantity $\bm{x}^{\text{T}}P_{\bm{Z}^{\perp}}\bm{x}$ in terms of $k,n$ and $p$, or equivalently, the decay rate of $\Var \hat{\beta}$.

	% Recall that $n$ is the total number of subjects (users), $k$ is the number of subjects to be assigned into the treatment  group and $p$ is the dimension of the covariates $\bm{Z}$. Ideally, we want to obtain a larger rate in order to have the variance decay even faster.

	We show that when $p$ is relatively small compared to $n$, the two strategies yield similar rates of $\Theta(n)$. In a high dimensional setting where $p$ is comparable to $n$, e.g., $p = n-1$, however, the random assignment is essentially useless in estimating $\beta$, in the sense that $\bm{x}^{\text{T}}P_{\bm{Z}^{\perp}}\bm{x}$ remains a constant in expectation as $n$ grows. Our proposed strategy, on the other hand, improves the rate to $\Theta(n)$ in this case, as long as $k/n$ is a constant. In Section \ref{random}, we discuss the randomized strategy. The optimal assignment algorithm is then considered in Section \ref{optimal}.

	%
	% We will first discuss the rate of $\bm{x}^{\text{T}}P_{\bm{Z}^{\perp}}\bm{x}$ where we randomly pick $k$ users into the treatment group. Then we will proceed with the optimality of the assignment strategy based on minimizing the variance of the estimator and an algorithm to approximate the optimal assignment.

	%\vspace{-0.1in}

	\section{RANDOM ASSIGNMENT}\label{random}

	Random assignment has been extensively studied in literature, mainly because it balances the covariates in two groups and eliminate the influence of confounders~\cite{causalreview1}. For our model in \eqref{prob:problem}, random assignment means that a subset of $k$ $x_i$'s are chosen at random and assigned a value $1$. Theorem \ref{theorem1} quantifies the rate of the increase of $\bm{x}^{\text{T}}P_{\bm{Z}^{\perp}}\bm{x}$.

	\begin{theorem}\label{theorem1}
		Random assignment achieves a rate of $\Theta$($(n-p)\frac{k(n-k)}{n^2}$). If $\frac{k}{n} = \rho$ is a constant, then this rate is $\Theta$($n-p$).
		\QEDA
	\end{theorem}

	Before proving Theorem \ref{theorem1}, we discuss the scaling rate under the setting when $k/n=\rho$ is a constant. In practice, this is the regime of interest since it is reasonable to suppose that a fraction (e.g. 10\%) of users receives DR signals. In this case, the rate achieved by random assignment is $\Theta(n-p)$. This rate is $\Theta$($n$) when $p$ is relatively small compared to $n$. However, when $p$ is large, e.g., $p = n-1$, then this rate becomes $\Theta$($1$). This rate is not desirable as it indicates that the variance of the estimator is not decaying with $n$ even when $n$ is large. Thus we would like to design an assignment strategy which yields an estimator that still possesses a relatively good performance even when $p$ is very close to $n$. In the next section, we show that with optimal assignment, we achieve the optimal rate $\Theta$($n$) when $\frac{k}{n} = \rho$. The proof of Theorem \ref{theorem1} follows.
	% % From Theorem \ref{theorem1}, we see that when $k = \log{n}$, then \eqref{eq:randomvar} can be written as:
	% \begin{equation}
	% (n-p)\frac{\log{n}(n-\log{n})}{n^2} = O(\frac{(n-p)\log{n}}{n}).
	% \end{equation}
	% Now supA more interesting case is when we replace $ \frac{k}{n} = \rho $ into $(n-p)\frac{k(n-k)}{n^2}$, then \eqref{eq:randomvar} can be written as:
	% \begin{equation}
	% (n-p) \frac{k(n-k)}{n^2} = (n-p) \rho (1-\rho) = O(n-p)
	% \end{equation}

	% This rate is O($n$) when $p$ is relatively small compared to $n$. This is true in a low dimensional setting. However, when $p$ is large, i.e., $p = n-1$, then this rate becomes O($1$). This rate is not desirable as it indicates that the variance of the estimator is not decaying with $n$ even when $n$ is large. Thus we would like to design an assignment strategy which yields an estimator that still possesses a relatively good performance even when $p$ is very close to $n$. In the next section, we will actually show that with optimal assignment, we achieve the optimal rate O($n$) when $\frac{k}{n} = \rho$.
	\begin{proof}
		We consider a random assignment where $\text{Pr} \{x_i=1\}=\frac{k}{n}$. Then the rate becomes:
		\begin{equation}\label{eq:randomvar}
		\begin{aligned}
		& \E \mathrm{tr}\{\bm{x}^{\text{T}}(I-\bm{Z}(\bm{Z}^{\text{T}}\bm{Z})^{-1}\bm{Z}^{\text{T}})\bm{x}\} \\
		= & k - \E \mathrm{tr} \{\bm{x}^{\text{T}}\bm{Z}(\bm{Z}^{\text{T}}\bm{Z})^{-1}\bm{Z}^{\text{T}}\bm{x}\}  \\
		\overset{(a)} = & k - \mathrm{tr} \{\bm{Z}(\bm{Z}^{\text{T}}\bm{Z})^{-1}\bm{Z}^{\text{T}} \E  \bm{x} \bm{x}^{\text{T}} \}\\
		\overset{(b)} = & k - \mathrm{tr} \{ \bm{Z}(\bm{Z}^{\text{T}}\bm{Z})^{-1}\bm{Z}^{\text{T}} \E (\tilde{\bm{x}} + \frac{k}{n})(\tilde{\bm{x}} + \frac{k}{n})^{\text{T}} \}\\
		\overset{(c)} = & k -   \mathrm{tr} \{  \bm{Z}(\bm{Z}^{\text{T}}\bm{Z})^{-1}\bm{Z}^{\text{T}} ( \frac{k(n-k)}{n^2}I + 0 + \frac{k^2}{n^2} \bm{1} \bm{1}^{\text{T}}) \} \\
		\overset{(d)} = & k - p\frac{k(n-k)}{n^2} - \frac{k^2}{n^2} n\\
		= & (n-p) k (n-k)/n^2, \\
		\end{aligned}
		\end{equation}
		where ($a$) follows from linearity and cyclic permutation of the trace operator; ($b$) follows from defining $\tilde{\bm{x}} = \bm{x} - \frac{k}{n}$; ($c$) follows from multiplying out each terms inside $(\tilde{\bm{x}} + \frac{k}{n})(\tilde{\bm{x}} + \frac{k}{n})^{\text{T}} $ and using the fact that each element in $\tilde{\bm{x}}$ has a zero mean and a variance as $\frac{k(n-k)}{n^2}$; ($d$) follows from $\bm{Z}(\bm{Z}^{\text{T}}\bm{Z})^{-1}\bm{Z}^{\text{T}}$ being a projection matrix onto $\bm{Z}$. Using the fact that the eigenvalues of a projection matrix are either 0 or 1 and $\bm{Z}$ has rank $p$ with probability one, then the trace of $\bm{Z}(\bm{Z}^{\text{T}}\bm{Z})^{-1}\bm{Z}^{\text{T}}$ is $p$ with probability one. In addition, from Lemma \ref{lemma0}, it is shown that $\bm{Z}(\bm{Z}^{\text{T}}\bm{Z})^{-1}\bm{Z}^{\text{T}}\bm{1}= \bm{1}$ if $\bm{Z}$ contains one column as intercept, so that the trace of  $\bm{Z}(\bm{Z}^{\text{T}}\bm{Z})^{-1}\bm{Z}^{\text{T}}\bm{1}\bm{1}^{T} = n$, which completes the equality in $(d)$.
		% is due to the fact that expectation and trace are interchangeable since they are linear operators. Equality ($a$) also follows from the fact that trace operator obeys cyclic permutations.
		% Now let  which has zero mean, then we obtain $(b)$. Equality $(c)$ naturally follows when we multiply out each terms inside $(\tilde{\bm{x}} + \frac{k}{n})(\tilde{\bm{x}} + \frac{k}{n})^{\text{T}} $. We also use the fact that each element in $\tilde{\bm{x}}$ has a zero mean and a variance as $\frac{k(n-k)}{n^2}$. What is more, note that $\bm{Z}(\bm{Z}^{\text{T}}\bm{Z})^{-1}\bm{Z}^{\text{T}}$ is a projection matrix onto $\bm{Z}$ and the eigenvalues of a projection matrix are either 0 or 1. Given that $\bm{Z}$ has rank $p$ with probability one, then the trace of $\bm{Z}(\bm{Z}^{\text{T}}\bm{Z})^{-1}\bm{Z}^{\text{T}}$ is $p$ with probability one. This explains the second term in equality (d).
	\end{proof}

	\begin{lemma}\label{lemma0}
		If $\bm{Z}$ is a $n$ by $p$ matrix (where $p<n$) with one column which contains all ones, then $\bm{Z}(\bm{Z}^{\text{T}}\bm{Z})^{-1}\bm{Z}^{\text{T}}\bm{1}= \bm{1}$.
		\QEDA
	\end{lemma}
	\begin{proof}
		Note that $I - \bm{Z}(\bm{Z}^{\text{T}}\bm{Z})^{-1}\bm{Z}^{\text{T}}$ is the projection matrix which is orthogonal to $\bm{Z}^{\text{T}}$, we then have the following:
		\begin{equation}\label{nulleq}
		\bm{Z}^{\text{T}}(I - \bm{Z}(\bm{Z}^{\text{T}}\bm{Z})^{-1}\bm{Z}^{\text{T}}) = \bm{0},
		\end{equation}
		where $\bm{0}$ is a zero vector that has length $n$.

		Note that $\bm{Z}$ has one column as the intercept, which suggests that $\bm{Z}^{\text{T}}$ has one row where each element takes value one. Since the equality in \eqref{nulleq} holds for every row, we then have:

		\begin{equation}
		\bm{1}^{\text{T}}(I - \bm{Z}(\bm{Z}^{\text{T}}\bm{Z})^{-1}\bm{Z}^{\text{T}}) = 0,
		\end{equation}
		which indicates that $\bm{1}^{\text{T}} \bm{Z}(\bm{Z}^{\text{T}}\bm{Z})^{-1}\bm{Z}^{\text{T}} = \bm{1}^{\text{T}}$.

	\end{proof}

	\section{OPTIMAL ASSIGNMENT}\label{optimal}
	Instead of being randomly assigned into DR programs, users can be optimally allocated to either the treatment group or the control group depending on their covariate information, in order to obtain the best estimator of $\beta$. Mathematically speaking, we optimally assign each $x_i$ to be 0 or 1, in order to minimize the variance of the estimator $\hat{\beta}$. This optimization problem is:
	\begin{equation}\label{prob:original0}
	\begin{aligned}
	& \underset{\hat{\bm{x}}}{\text{maximize}}
	&& \bm{x}^{\text{T}}P_{\bm{Z}^{\perp}}\bm{x}\\
	& \text{subject to}
	&&  \sum_{i=1}^{n}{x}_i = k\\
	&&& {x}_i \in \{1,0\}.
	\end{aligned}
	\end{equation}

	We first discuss the upper bound on the quantity $\bm{x}^{\text{T}}P_{\bm{Z}^{\perp}}\bm{x}$ (which signifies the lower bound for $\Var \hat{\beta}$). We show that it is O($n$).
	% This indicates that no possible strategy can obtain any higher rate than O($n$).
	Then we establish that under the regime of $k/n=\rho$, there exist algorithms that achieve a rate that meets the upper bound of $O(n)$.
	% We then discuss the rate with the assignment by the proposed optimization problem. We finally show an interesting result that the proposed assignment based on variance minimization in \eqref{prob:original0} can achieve the optimal rate when $\frac{k}{n}$ is a constant.

	\subsection{Optimal Rate}

	Before proceeding on analyzing the rate obtained by the proposed strategy, we first discuss the upper bound on the rate of $\bm{x}^{\text{T}}P_{\bm{Z}^{\perp}}\bm{x}$.
	%which means that we should not expect a better rate than this bound with any assignment strategies.

	\begin{proposition}\label{proposition0}
		No assignment can achieve a better rate that O$(n)$~\cite{ABtest}.\QEDA
	\end{proposition}
	\begin{proof}
		The basic idea is to derive the Fisher information with the linear regression model in \eqref{prob:problem}. The inverse of the Fisher information provides a lower bound for the variance of the estimator obtained by least squares and thus an upper bound for the quantity $\bm{x}^{\text{T}}P_{\bm{Z}^{\perp}}\bm{x}$. For more details, please refer to Proposition 1 in~\cite{ABtest}.
	\end{proof}
	In the next subsection we will show that when $\frac{k}{n} = \rho$ which is a constant, we achieve this upper bound.

	\subsection{Achievability of Optimal Rate}
	We first present the main result of this section. We assume that each element of $\bm{Z}_{1:p-1}$ (excluding the intercept column) is drawn independently from a standard Gaussian distribution. This assumption will facilitate the calculation of the main result shown in Theorem \ref{theorem2}. The algorithm associated with Theorem \ref{theorem2} is presented in Algorithm \ref{algo0}.
	\begin{theorem}\label{theorem2}
		Recall that the rate is the growing rate of the inverse of the variance introduced in \eqref{eq:varbeta}. This rate from optimal assignment is of $\Theta$($\frac{k^2log(\frac{n}{k})}{n}$), which is independent of the dimension of covariates. More specifically, when $\frac{k}{n} = \rho$ is a constant, then this rate is  linear rate, i.e., $\Theta$($n$).
		\QEDA
	\end{theorem}
		\begin{algorithm}
			%\DontPrintSemicolon % Some LaTeX compilers require you to use \dontprintsemicolon    instead
			\KwIn{Covariates $\bm{Z}$.}
			\KwOut{Rate of optimal assignment and the corresponding optimal assignment strategy when $p = n-1$.}
			Reduce the optimization problem in \eqref{prob:original0} to \eqref{prob:p=n-1} using Lemma \ref{increasing}.

			Compute the null space of $\bm{Z}_{n,n-1}^{\text{T}}$, denote it by $\bm{y}$. Each element of $\bm{y}$ should independently follow a standard Gaussian distribution, according to Lemma \ref{Gaussiannull}.

			Find the lower bound for the $k^{th}$ largest element in $\bm{y}$ (suppose that this element is non negative). This lower bound is shown in Lemma \ref{lemma2}.

			The optimal value of the objective function in \eqref{prob:p=n-1} is at least $\frac{k^2}{n}$ times this lower bound. The rate of this optimal value is stated in Theorem \ref{theorem2}. The optimal assignment is to assign those $x_i$'s corresponding to the $k$ largest $y_i$'s in $\bm{y}$ to be 1's and the rest to be 0's.

			\caption{Procedures to obtain the rate shown in Theorem \ref{theorem2}.}
			\label{algo0}
		\end{algorithm}

	Before proving Theorem \ref{theorem2}, we first show in Lemma \ref{increasing} that the worst case scenario for the rate is when $p = n - 1$. This scenario provides a minimum on the quantity $\bm{x}^{\text{T}}P_{\bm{Z}^{\perp}}\bm{x}$ for every $p$ where $p<n$, which provides a maximum for $\Var \hat{\beta}$ for every $p < n$. Thus if we can show in Theorem \ref{theorem2} that in the worst case scenario where $p = n - 1$, the growing rate of quantity $\bm{x}^{\text{T}}P_{\bm{Z}^{\perp}}\bm{x}$ is $\Theta(n)$ when $\frac{k}{n} = \rho$ is a constant, then this rate holds for all $p$ where $p<n-1$.

	\begin{lemma}\label{increasing}
		$\Var \hat{\beta}$ is increasing in $p$. Consequently, if $p$ = $n-1$, the estimator yields the worse case performance \cite{ABtest}:
		\begin{equation}
		\inf_{1 \leq p < n}  \frac{\bm{x}^{\text{T}}P_{\bm{Z}_{n,p}^{\perp}}\bm{x}}{n} = \frac{\bm{x}^{\text{T}}P_{\bm{Z}_{n,n-1}^{\perp}}\bm{x}}{n}
		\end{equation}
		\QEDA
	\end{lemma}

	\begin{proof}
		This is a general result about linear estimation and the interested reader can refer to Lemma 5 in~\cite{ABtest}.
	\end{proof}

	When $p = n-1$, the rank of $\bm{Z}_{n,n-1}$ is one with probability one, thus we write $P_{\bm{Z}_{n,p}^{\perp}} = \frac{\bm{y}\bm{y}^{\text{T}}}{||\bm{y}||^2}$, where $\bm{y}$ is in the null space of $\bm{Z}_{n,n-1}^{\text{T}}$ \cite{ABtest}, i.e., $\bm{Z}^{\text{T}}_{n,n-1}\bm{y} = \bm{0}$. Based on this observation, $\bm{x}^{\text{T}}P_{\bm{Z}^{\perp}}\bm{x}$ is written into a simpler form as $\frac{\bm{x}^{\text{T}}\bm{y}\bm{y}^{\text{T}}\bm{x}}{||\bm{y}||_2^2} = \frac{(\bm{y}^{\text{T}}\bm{x})^2}{||\bm{y}||^2}$. The problem is then to maximize $\frac{(\bm{y}^{\text{T}}\bm{x})^2}{||\bm{y}||^2}$ under the constraint that we only get to assign $k$ $x_i$'s to be 1's and the rest to be 0's. The optimization problem is:
	\begin{equation}\label{prob:p=n-1}
	\begin{aligned}
	& \underset{\hat{\bm{x}}}{\text{maximize}}
	&& \frac{(\bm{y}^{\text{T}}\bm{x})^2}{||\bm{y}||^2}\\
	& \text{subject to}
	&&  \sum_{i=1}^{n}{x}_i = k\\
	&&& {x}_i \in \{1,0\}.
	\end{aligned}
	\end{equation}

	To solve the optimization problem in \eqref{prob:p=n-1}, we need to find $k$ $y_i$'s in $\bm{y}$ such that their sum is maximized, where $y_i$ is the $i^{th}$ element of the vector $\bm{y}$. We observe that it actually suffices to provide a lower bound on this maximum sum to prove the rate.

	To provide this lower bound we need to know the structure of $\bm{y}$. We then show in Lemma \ref{Gaussiannull} that if $\bm{y}$ is in the null space of $\bm{Z}_{n,n-1}$, then each $y_i$ can be constructed to be drawn from an i.i.d. standard Gaussian distribution. Based on this observation, the problem is further reduced to find the lower bound on the $k^{th}$ largest $y_i$, assuming that $2k$ is smaller than $n$ to ensure that with overwhelming probability the $k^{th}$ largest $y_i$ is non negative. Let us refer to this statistic as the $(n-k+1)^{th}$ order statistic of $\bm{y}$ and denote it by $y_{(n-k+1)}$ such that $y_{(1)} \leq y_{(2)} \leq \dots \leq y_{(n)}$. We present a lower bound for $y_{(n-k+1)}$ in Lemma \ref{lemma2} when $k$ is smaller than $\frac{n}{2}$. This lower bound facilitates the final proof for Theorem \ref{theorem2}. The proof of Lemma \ref{Gaussiannull} and Lemma \ref{lemma2} is left in the appendix.

	\begin{lemma}\label{Gaussiannull}
		The basis of the null space of $\bm{Z}_{n,n-1}^{\text{T}}$ can be constructed as an i.i.d. standard Gaussian vector with length $n$, when each element of $\bm{Z}_{1:n-2}$ is independently drawn from a standard Gaussian distribution and the last column of $\bm{Z}$ is an all one column.
		\QEDA
	\end{lemma}

	\begin{lemma}\label{lemma2}
		Let $\bm{y}$ satisfies $\bm{Z}_{n,n-1}^{\text{T}}\bm{y} = 0$. If each $y_i$ is independent and follows standard Gaussian distribution, then
		$\E y_{(n-k+1)} \geq C\sqrt{\log{\frac{n}{k}}}$, where C is a positive constant and $\frac{k}{n} < \frac{1}{2}$.
		\QEDA
	\end{lemma}

	Now we can use the introduced lemmas to prove Theorem \ref{theorem2}. A summary is presented in Algorithm \ref{algo0}, illustrating the procedures to obtaining the rate stated in Theorem \ref{theorem2} using the proposed lemmas. This algorithm also provides the optimal assignment strategy when $p = n - 1$.

	\begin{proof}[Proof of Theorem \ref{theorem2}]

		We will focus on the case when $p = n-1$ since it provides the worst case rate for every $p < n$, as stated in Lemma \ref{increasing}.

		From lemma \ref{Gaussiannull}, we know that $\bm{y} \sim N(0, \text{I}_{n})$, we then obtain the following results:
		\begin{equation} \label{eq:optimalvar}
		\begin{aligned}
		& \underset{\bm{x}, x_i \in \{1,0\}, \sum x_i = k}{\text{max}}\E \frac{(\bm{y}^{\text{T}}\bm{x})^2}{||\bm{y}||_2^2} \\
		& \geq \E \frac{\{(y_{(n)}+...+y_{(n-k+1)})^2\}}{||\bm{y}||_2^2}\\
		&  \overset{(a)} = \frac{\E \{(y_{(n)}+...+y_{(n-k+1)})^2\}}{\E ||\bm{y}||_2^2} + O(\frac{1}{n})\\
		& \geq \frac{\E \{(ky_{(n-k+1)})^2\}}{n} + O(\frac{1}{n})\\
		& = k^2 \frac{\E \{y^2_{(n-k+1)}\}}{n} + O(\frac{1}{n})\\
		& \overset{(b)} \geq k^2 \frac{(\E y_{(n-k+1)})^2}{n} + O(\frac{1}{n})\\
		& \overset{(c)}\geq k^2 \frac{C^2log\frac{n}{k}}{n} + O(\frac{1}{n}),\\
		\end{aligned}
		\end{equation}
		where ($a$) is based on the multivariate delta method~\cite{deltamethod}, ($b$) comes from Jensen's inequality and ($c$) is based on Lemma \ref{Gaussiannull} and Lemma \ref{lemma2}.

		Specifically, if $\frac{k}{n} = \rho$, then \eqref{eq:optimalvar} can be written as:
		\begin{equation}
		k^2 \frac{C^2log\frac{n}{k}}{n} = C^2\rho^2\log{(\rho^{-1}})n = \Theta(n).
		\end{equation}

		Another interesting case is when $k = \log{n}$, \eqref{eq:optimalvar} can be written as:
		\begin{equation}
		k^2 \frac{C^2\log{\frac{n}{k}}}{n} = C^2(\log{n} - \log{\log{n}}) = \Theta(\log{n}).
		\end{equation}
	\end{proof}

	As can be seen from Theorem \ref{theorem2}, we will obtain the optimal rate $\Theta(n)$ by replacing $\frac{k}{n}$ as a constant. This indicates that with optimal assignment, the estimation variance will indeed decay with $n$ instead of being a constant as shown in Theorem \ref{theorem1}, when the dimension of the covariates $p$ is comparable to $n$. This is very interesting as it indicates that even in a high dimensional setting, the variance of the estimator will decay optimally by solving the variance minimization problem.

		\section{A TRACTABLE ALTERNATIVE}\label{sec:SDP}
		Algorithm \ref{algo0} provides a simple way to find the optimal assignment when $p = n - 1$. When $p$ is less than $n-1$, this algorithm cannot be applied. In fact, the optimization problem in \eqref{prob:original0} is a nonconvex quadratic optimization problem which can be NP-hard. In this section, we present a tractable approximate algorithm by relaxing the original combinatorial optimization problem into and semidefinite program. We then demonstrate that this SDP problem approximates the original problem with a performance ratio that is better than $\frac{k}{n}$ when $\frac{k}{n}$ is in the range of $(0.2, 0.9995)$. This procedure follows the results established in \cite{SDP}.

		% Now we know that optimizing the assignment can improve the performance of the estimator, we need to resort to some tractable algorithms to solve the original optimization problem (which is hard). In this section we present a tractable alternative that approximates the original combinatorial problem into a SDP formulation.  This rate is proved and achieved in \cite{SDP}.

		We first revisit the original variance minimization problem in \eqref{prob:original0} and transform $\bm{x}$ into $\bm{x} = (\hat{\bm{x}} + 1)/2$. Denote each element in $\hat{\bm{x}}$ as $\hat{x}_i$, then each $\hat{x}_i$ takes value in $\{-1,1\}$. Therefore the variance minimization problem is written as:
		\begin{equation}\label{prob:original}
		\begin{aligned}
		& \underset{\hat{\bm{x}}}{\text{maximize}}
		&& \frac{1}{4}(\hat{\bm{x}}+1)^{\text{T}}P_{\bm{Z}^{\perp}}(\hat{\bm{x}}+1)\\
		& \text{subject to}
		&& \sum_{i=1}^{n}\hat{x}_i = 2k-n\\
		&&& \hat{x}_i \in \{1,-1\}.
		\end{aligned}
		\end{equation}

		This is a Dense-$k$-Subgraph (DSP) problem with existence of self edges. To illustrate this, let element at row $i$ and column $j$ in matrix $P_{\bm{Z}^{\perp}}$ be denoted as edge weight $w_{ij}$ (except that $w_{ii}$ is half of the value on the diagonals) associated with vertex $i$ and vertex $j$, then \eqref{prob:original} is trying to find a set of $k$ vertices such that the sum of edge weights induced by these vertices are maximized. Since the problem presented in \eqref{prob:original} contains binary variables, we relax this problem into a SDP formulation:
		\begin{equation}\label{eq:SDP}
		\begin{aligned}
		&  \underset{\bm{X}, \hat{\bm{x}}}{\text{maximize}}
		& &   \frac{1}{4}\sum_i\sum_j w_{ij}(1+\hat{x}_i+\hat{x}_j+X_{ij})\\
		& \text{subject to}
		&&  \sum_i \hat{x}_i = 2k-n\\
		&&&  X_{ii} = 1\\
		&&& \sum_i\sum_j X_{i,j} = (2k-n)^2 \\
		&&& \begin{bmatrix}
		1 & \hat{\bm{x}}^{\text{T}}\\
		\hat{\bm{x}} & \bm{X}
		\end{bmatrix} \succeq 0.
		\end{aligned}
		\end{equation}

		The original problem in \eqref{prob:original} is hard, we will only obtain a surrogate solution in polynomial time. We thus adopt Algorithm \ref{algo} to obtain an approximate solution from SDP formulation. Let us denote this solution by $\hat{\bm{x}}^*$ based on Algorithm \ref{algo}. The performance of the approximation from SDP is evaluated by the performance ratio $r$ which satisfies:
		\begin{equation}\label{ratioex}
		\E \frac{1}{4}(\hat{\bm{x}}^*+1)^{\text{T}}P_{\bm{Z}^{\perp}}(\hat{\bm{x}}^*+1)
		\geq rw^*.
		\end{equation}

		Here the randomness in $\hat{\bm{x}}^*$ is introduced by the random rounding procedure shown in Algorithm \ref{algo} and $w^*$ is the optimal value of the objective function shown in \eqref{prob:original}.  Performance ratio $r$ can be used to quantify how close the solution from Algorithm \ref{algo} is to the optimal solution by solving the original hard problem.

		There exists a fruitful line of work on the approximation algorithms using either greedy algorithm or LP/SDP relaxation for DSP problems~\cite{KortsarzPeleg12, KortsarzPeleg8, Asahiro3, FeigLangber6, FeigSeltser7, SrivastavWolf16,YeZhang19, SDP, DkP14}. Most recent research has improved the performance ratio to $O(n^{-\frac{1}{4}+\epsilon})$ with LP relaxation in~\cite{DkP14}. However, if $\frac{k}{n}$ is not decaying with $n$, i.e., a constant, then this ratio is not desirable since it is decreasing in $n$. The authors in~\cite{SDP} propose an improved performance ratio that is better than $\frac{k}{n}$ for a wide range of $\frac{k}{n}$. We will adopt the approximation procedure in~\cite{SDP} and argue that the performance ratio is valid in our case here as well. In the following, we will first present the general algorithm and then show that the performance ratio in~\cite{SDP} is still applicable in our case.

		The approximation procedure in~\cite{SDP} is presented in Algorithm \ref{algo}, including three main procedures:
		\begin{itemize}
			\item Solve SDP problem in \eqref{eq:SDP} (step 1). After this procedure we can obtain the optimal continuous solution for \eqref{eq:SDP} and the optimal value of the objective function is denote by $w^{SDP}$.

			\item Construct initial $S$, where $S$ represents the initial subgraph and is a set of indices (step 2 through step 4). The $\hat{x}_i$'s take value 1 such that $i \in S$ and the rest -1. Let us denote them by $\hat{\bm{x}}^{0}$. The value of the objective function in \eqref{eq:SDP} is written as $w(S) = \frac{1}{4}\sum_i\sum_j w_{ij}(1+\hat{x}_i^{0}+\hat{x}_j^{0}+\hat{x}_i^{0}\hat{x}_j^{0}) = \frac{1}{4}(\hat{\bm{x}}^{0}+1)^{\text{T}}P_{\bm{Z}^{\perp}}(\hat{\bm{x}}^{0}+1)$. Here $w({S})$ is the total weights of edges in the subgraph induced by ${S}$. At this point the cardinality of $S$ is not necessarily $k$.

			\item Resize $S$ to $\tilde{S}$ such that $\tilde{S}$ contains exactly $k$ vertices (step 5 through step 16). The final assignment of $\hat{x}_i$'s is that $\{\hat{x}_i = 1, \hat{x}_j = -1| i \in \tilde{S}, j \notin \tilde{S}\}$. Let us denote them by $\hat{\bm{x}}^{*}$. The value of the objective function is $w(\tilde{S}) = \frac{1}{4}(\hat{\bm{x}}^{*}+1)^{\text{T}}P_{\bm{Z}^{\perp}}(\hat{\bm{x}}^{*}+1)$ and is the total weights of the edges induced by $\tilde{S}$.
		\end{itemize}

		\begin{algorithm}
			%\DontPrintSemicolon % Some LaTeX compilers require you to use \dontprintsemicolon    instead
			\KwIn{$w_{ij}, k, n$}
			\KwOut{$S$}
			Solve SDP in \eqref{eq:SDP}, obtain $\bm{X}, \hat{\bm{x}}$.\\

			%\renewcommand{\labelenumi}{(\Roman{enumi})}
			%\begin{enumerate}[noitemsep,nolistsep]
			%\item
			Construct $\bar{\bm{X}} =  \begin{bmatrix}
			1 & \hat{\bm{x}}^{\text{T}}\\
			\hat{\bm{x}} & \bm{X}
			\end{bmatrix}$.
			\\
			Construct covariance matrix $Y = \theta \bar{\bm{X}} + (1-\theta)\bm{P}$,
			%\[
			where:
			\linebreak
			$0\leq\theta \leq 1$,
			\linebreak
			$\bm{P} = \begin{bmatrix}
			1 & \chi & \chi & \dots  & \chi \\
			\chi & 1 & \chi^2 & \dots  & \chi^2 \\
			\vdots & \vdots & \vdots & \ddots & \vdots \\
			\chi & \chi^2 & \dots & \chi^2  & 1
			\end{bmatrix}$,
			\linebreak
			$\chi = 2\frac{k}{n}-1$.
			\\
			%\]
			%\item
			Generate $\bm{u} \sim N(0,\bm{Y})$,
			$\tilde{\bm{x}} = sign(\bm{u})$, $S = \{i \geq 2:\tilde{x}_i = \tilde{x}_1\}$.\\

			%\end{enumerate}
			Let $\tilde{S} = S$.\\
			\eIf{$|\tilde{S}| = k$}{Output $\tilde{S}$}{
				\eIf{$|\tilde{S}|<k$}{arbitrarily add $k-|\tilde{S}|$ nodes into $\tilde{S}$. Output $\tilde{S}$.}{
					\While{$|\tilde{S}|>k$}{
						\For{each $i \in \tilde{S}$ }{
							$\eta_i = \sum_{j \in \tilde{S}}{w_{ij}}$.
						}
						Rearrange $\tilde{S} = \{i_1, i_2, \dots, i_{|\tilde{S}|}\}$, where $\eta_{i_1}\geq\eta_{i_2}\geq \dots \geq \eta_{i_{|\tilde{S}|}}$. Remove $i_{|\tilde{S}|}$ and reset $\tilde{S} = \{i_1, i_2, \dots,i_{|\tilde{S}|-1}\} $.

					}
					Output $\tilde{S}$.
				}

			}

			\caption{Approximation algorithm with SDP relaxation in \eqref{eq:SDP}, adopted from~\cite{SDP}.}
			\label{algo}
		\end{algorithm}

		We now demonstrate that Algorithm \ref{algo} achieves the performance ratio shown in Proposition \ref{propositionratio}.

		\begin{proposition}\label{propositionratio}
			The performance ratio $r$ from Algorithm \ref{algo} defined as:
			\begin{equation}
			\E w(\tilde{S}) \geq r w^*
			\end{equation}
			satisfies the conditions presented in Proposition 2 in \cite{SDP} and is plotted in Fig. \ref{figratio}, where $w^*$ is the optimal value of the objective function in problem \eqref{prob:original}. When $k$ is large, this ratio is better than either O($n^{-(\frac{1}{3}-\epsilon)}$) or O($n^{-(\frac{1}{4}-\epsilon)}$) obtained by LP relaxation.
			\QEDA
		\end{proposition}

		\begin{figure}[!t]
			\centering
			\includegraphics[width = 0.8\columnwidth]{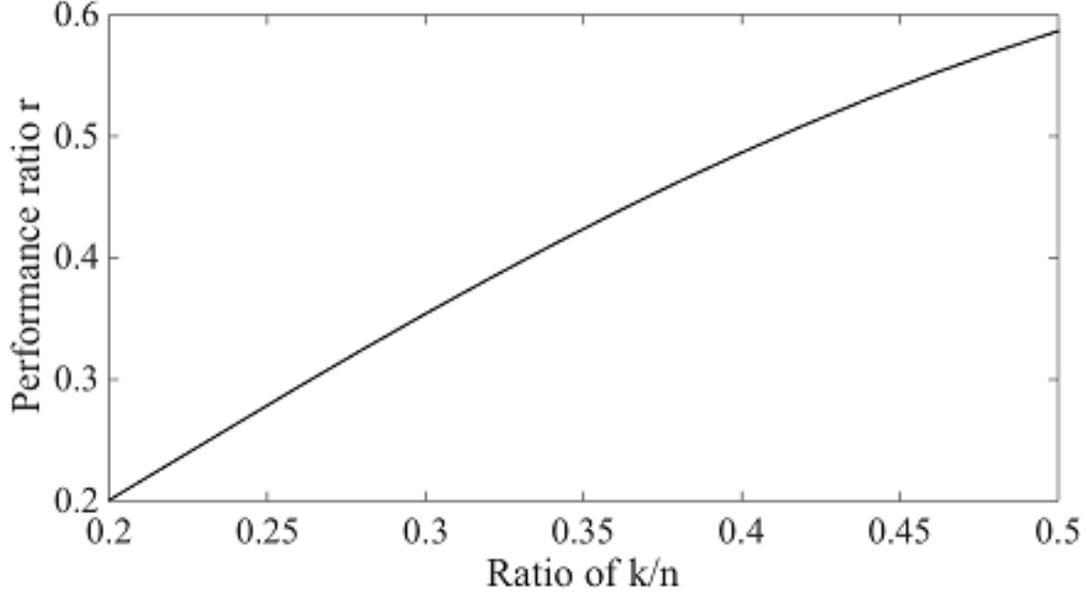}
			\caption{ Performance ratio from Algorithm \ref{algo} for $\frac{k}{n} \in [0.2, 0.5]$. }
			\label{figratio}
		\end{figure}

		As can be seen from Proposition \ref{propositionratio}, the performance ratio $r$ quantifies the gap between the approximated solution obtained by Algorithm \ref{algo} and the optimal solution from the original problem shown in \eqref{prob:original}. This ratio $r$ is the direct result from the random rounding procedure and the resizing procedure shown in Algorithm \ref{algo}. The ratios associated with these two procedures are presented in \eqref{bounda} and \eqref{boundb}:

		\begin{equation}\label{bounda}
		\E w(S) \geq \alpha w^*,
		\end{equation}
		and
		\begin{equation}\label{boundb}
		w(\tilde{S}) \geq \xi w(S).
		\end{equation}

		In \cite{SDP}, the authors well define the parameter $\alpha$ (which depends on $k$ and $n$) and $\xi$ (depends on $k$ and $S$ obtained from the random rounding procedure) so that the performance ratio $r$ satisfies Proposition \ref{propositionratio} where there are no self edges, i.e., $w_{ii} = 0$. However, in our case, the $w_{ii}'s $ are the diagonal elements of $\bm{P}_{\bm{Z}^{\perp}}$ and they are not necessarily zero, so the graph in our case contains non trivial self edges. If we show that the presence of self edges does not change the values of $\alpha$ and $\xi$, then Proposition \ref{propositionratio} naturally holds in our case as well.

		Let us first discuss \eqref{bounda}. From \cite{SDP}, parameter $\alpha$ does not depend on whether there are self edges in the graph, so \eqref{bounda} directly applies.

		Next we need to check if the same $\xi$ applies at the presence of self edges. When there are no self edges, i.e., $w_{ii} = 0$, the authors in \cite{SDP} show that $\xi = \frac{k(k-1)}{|S|(|S|-1)}$ when $|S|>k$ and $\xi = 1$ otherwise. We show that the same condition for $\xi$ holds even with the presence of non negative self edges, as stated in Lemma \ref{selfedge}.

		\begin{lemma}\label{selfedge}
			Let $S$ and $\tilde{S}$ be obtained from random rounding procedure and resizing procedure in Algorithm \ref{algo} from a graph with non negative self edges, i.e., $w_{ii} \geq 0$. Then we have:
			\[
			w(\tilde{S})=
			\begin{cases}
			\frac{k(k-1)}{|S|(|S|-1)}w(S),& \text{if } |S| > k\\
			1,              & \text{otherwise}
			\end{cases}
			\]
			\QEDA
		\end{lemma}

	Proof of Lemma \ref{selfedge} is given in the appendix. Lemma \ref{selfedge} validates the $\alpha$ and $\xi$ with the presence of non negative self edges, thus the performance ratio $r$ stated in Proposition \ref{propositionratio} is valid. A list of performance ratio $r$'s for different values of $\frac{k}{n}$ is shown in~\cite{SDP} and is plotted in Fig. \ref{figratio}. This rate is satisfactory since it has a rate of $O(\frac{k}{n})$ but strictly larger than $\frac{k}{n}$ \cite{SDP}. It is better than O($n^{-\frac{1}{4}+\epsilon}$) when $\frac{k}{n}$ is not decaying faster than $O(n^{-\frac{1}{4}+\epsilon})$. In fact, as long as $\frac{k}{n}$ lies within some constant range, then $\frac{k}{n}$ can be seen as a varying constant thus is not decaying as a function of $n$.

	\section{SIMULATION}\label{simulation}
	In this section we show the comparison results between random assignment and optimal assignment. We simulate the covariates from two different distributions, i.e., Gaussian distribution and uniform distribution. We also validate the claim by simulated building data from \cite{sim}.

	\subsection{Gaussian Ensemble}

	We first generate the covariates as they are drawn from i.i.d. Gaussian ensemble, i.e., $N(\bm{0}, {I})$. We compare two cases where $n = 3k$ and $n = 5k$ in Fig. \ref{Gaussian}. Note that Fig. \ref{Gaussian} is shown in semilogarithmic plot where the $y$-axis has a logarithmic scale and the $x$-axis has a linear scale. In addition, we adopt the value of $\theta$ in \cite{SDP}, i.e., 0.9 for $n = 3k$ and 0.94 for $n = 5k$. We let $p = n -1$ to obtain the worst case performance.

	Besides SDP relaxation, we also simulate a greedy based assignment to maximize the weighted edges induced by $k$ vertices. The greedy assignment sequentially eliminates vertices and works as follows: we start with the original graph and a set containing all vertex. At each elimination step, the vertex with the least weighted edges are eliminated from the set until this set contains exactly $k$ vertices. This greedy algorithm is introduced in \cite{Asahiro3}.

	In addition, we use the result from branch and bound (upper bound) to serve as the reference in order to compare the random assignment and the proposed optimal assignment. The duality gap for branch and bound is set to be 0.05 for all $n$ when $n = 3k$. Due to computational complexity and time constraints, we set this gap to be around 0.25 when $n$ is big in the case of $n = 5k$.

	In Fig. \ref{Gaussian}, we see that the semilog plot on $\Var^{-1} \hat{\beta}$ ($\bm{x}^{\text{T}}P_{\bm{Z}_{n,n-1}^{\perp}}\bm{x}$) is growing with $n$ and similar to $\log n$. This suggests that $\Var^{-1} \hat{\beta}$ is linear in $n$, as we stated in Section \ref{optimal}. In addition, It is very close to the solution obtained by branch and bound, meaning that the result from SDP relaxation is close to the optimal solution in \eqref{prob:original}. The empirical performance ratio from SDP relaxation is shown in Table I.

	From Table I, we see that the performance ratio for $n$ between $10$ and $200$ is actually greater than $\frac{k}{n}$, which is even better than the theoretical bound in Proposition \ref{propositionratio}. In addition, if we decrease $k$ with respect to $n$, i.e., change $k$ from $\frac{n}{3}$ to $\frac{n}{5}$, then the performance ratio is reduced. This is due to the fact we need to do more eliminations during the resizing procedure in Algorithm \ref{algo} and the deterioration increases. However, if $\frac{k}{n}$ is well defined within a range, then this deterioration is controlled and does not change the statement in Proposition \ref{propositionratio}.

	On the other hand, the semilog plot on $\Var^{-1} \hat{\beta}$ from random assignment is a constant on average across different values of $n$, whether $k$ is small or large. This validates Theorem \ref{theorem1} as it states that the rate is $\Theta$($1$) when $p = n-1$. In this case, we cannot obtain an efficient estimator since the variance is not decaying even when $n$ is big.

	What is more, the greedy assignment does not provide a relatively good performance as well. It yields a constant $\Var \hat{\beta}$ as random assignment. This suggests although both the greedy algorithm and SDP relaxation are aimed to solve an optimization problem, the solution from SDP relaxation is much better and reliable.

	\begin{figure}[!t]
		\centering
		\includegraphics[width = 0.8\columnwidth]{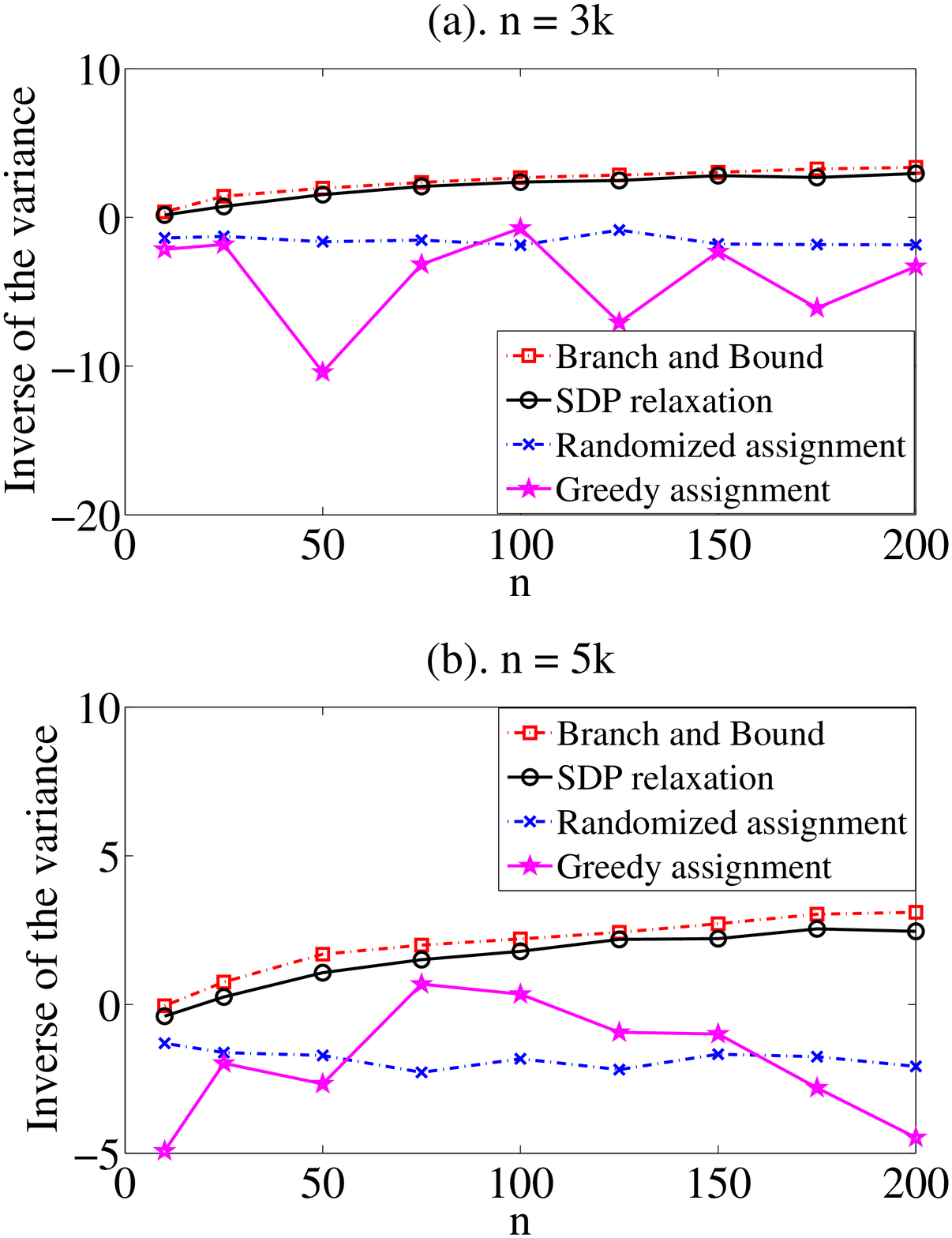}
		\caption{ Semilogarithmic plot on the inverse of $\Var \hat{\beta}$, i.e., $\bm{x}^{\text{T}}P_{\bm{Z}_{n,n-1}^{\perp}}\bm{x}$,  assuming Gaussian distribution of $\bm{Z}_{1:n-2}$. The $y$-axis is shown in a logarithmic scale and the $x$-axis is shown in a linear scale. Upper plot shows the rate when $n = 3k$, lower plot shows the corresponding rate when $n = 5k$. }
		\label{Gaussian}
	\end{figure}

	\begin{table}[!ht]\label{table1}
		\renewcommand{\arraystretch}{1.3}
		\caption{Empirical performance ratio for Gaussian ensemble with different values for $\frac{k}{n}$ and varying $n$.} %In the table, SLR stands for simple linear regression, MLR stands for multiple linear regression and MCM stands for modified covariate method.
		\centering
		\begin{tabular}{|c|c|c|c|c|}
			\hline
			\bfseries $\rho = \frac{k}{n}$ & \bfseries $n$ = 10 & \bfseries  $n$ = 50 & \bfseries  $n$ = 100 & \bfseries  $n$ = 200 \\
			%\hline
			%Case 11 & 317& 403 & <2e-16 & <2e-16 \\
			\hline
			$\frac{1}{3}$  & 0.8131 & 0.6486  &  0.7416 & 0.6588 \\
			\hline
			$\frac{1}{5}$  & 0.7026 & 0.5362 & 0. 6577 & 0.5273 \\
			\hline
		\end{tabular}
	\end{table}

	\subsection{Uniform Ensemble}
	Although we discuss the rate of quantity $\bm{x}^{\text{T}}P_{\bm{Z}_{n,n-1}^{\perp}}\bm{x}$ with respect to Gaussian ensemble, we also simulate the covariates where the elements are drawn from a uniform distribution in an interval [-1,1]. The results are shown in Fig. \ref{Uniform}. Fig. \ref{Uniform} is again a semilogarithmic graph. We again take $\theta$ to be 0.9 when $n = 3k$ and 0.94 when $n = 5k$. The duality gap is 0.05 when $n = 3k$. When $n = 5k$, the duality gap is around 0.1 to 0.2 for $n$ greater than 100 and is 0.25 for $n = 200$. The comparison result is shown in Fig. \ref{Uniform} and the performance ratio by SDP relaxation is shown in Table II.

	\begin{figure}[!t]
		\centering
		\includegraphics[width = 0.8\columnwidth]{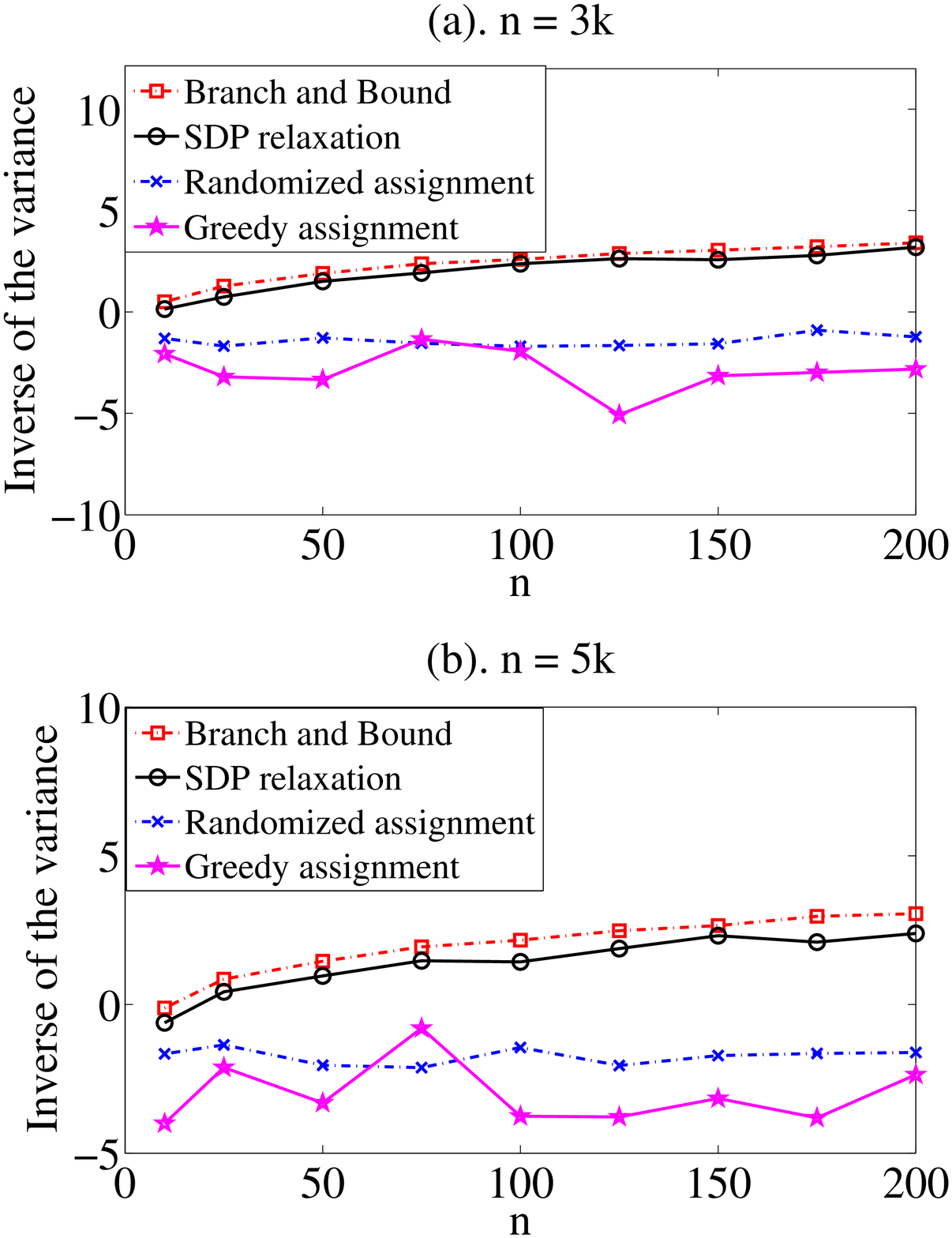}
		\caption{Semilogarithmic plot on the inverse of $\Var \hat{\beta}$, i.e., $\bm{x}^{\text{T}}P_{\bm{Z}_{n,n-1}^{\perp}}\bm{x}$,  assuming uniform distribution of $\bm{Z}_{1:n-2}$. Upper plot shows the rate when $n = 3k$, lower plot shows the corresponding rate when $n = 5k$.}
		\label{Uniform}
	\end{figure}

	\begin{table}[!ht]\label{table2}
		\renewcommand{\arraystretch}{1.3}
		\caption{Empirical performance ratio for uniform ensemble with different values for $\frac{k}{n}$ and varying $n$.} %In the table, SLR stands for simple linear regression, MLR stands for multiple linear regression and MCM stands for modified covariate method.
		\centering
		\begin{tabular}{|c|c|c|c|c|}
			\hline
			\bfseries $\rho = \frac{k}{n}$ & \bfseries $n$ = 10 & \bfseries  $n$ = 50 & \bfseries  $n$ = 100 & \bfseries  $n$ = 200 \\
			%\hline
			%Case 11 & 317& 403 & <2e-16 & <2e-16 \\
			\hline
			$\frac{1}{3}$  & 0.6961 & 0.6755  &  0.8068 & 0.8094 \\
			\hline
			$\frac{1}{5}$  & 0.6101 & 0.6113 & 0.4799 & 0.5145 \\
			\hline
		\end{tabular}
	\end{table}

	The observation from Fig. \ref{Uniform} and Table II is similar to the analysis in the case of Gaussian ensemble, that the solution obtained from SDP relaxation is still within a constant of the branch and bound solution for $10 \leq n \leq 200$. Again, greedy algorithm fails to find a solution close to that by branch and bound and performs as poorly as the random assignment. The performance ratio is again decreased when we decrease $k$, which indicates the similar deterioration occurred during the resizing procedure.

	\subsection{Building data}

	We also validate the claim in the paper by simulated building data obtained from \cite{sim}. Using this software, we can generate building covariates such as environment temperature, number of occupants, appliances scheduling, etc. The software outputs the energy consumption based on the covariates that we generate. The buildings vary from a small postal office, to a large commercial hotel, with different number of covariates involved in the modeling.

	For the purpose of this paper, we include significant covariates into the linear regression model and the number of those covariates are comparable to the number of users in the simulation. The values of the covariates are whitened by the covariance matrix so that the covariates have zero mean and unit variance \cite{CarlosEtAl2016}. The details of this procedure is given in the appendix. In addition, we fix $k = \frac{1}{3}n$ in the simulation.

	The simulation results are shown in Fig. \ref{building} and Fig. \ref{varying_d}.

	In Fig. \ref{building}, we observe that $\Var \hat{\beta}$ is decaying very fast with the optimal assignment strategy, whereas the variance stays unchanged if adopting a random assignment strategy.

	\begin{figure}[!t]
		\centering
		\includegraphics[width = 0.75\columnwidth]{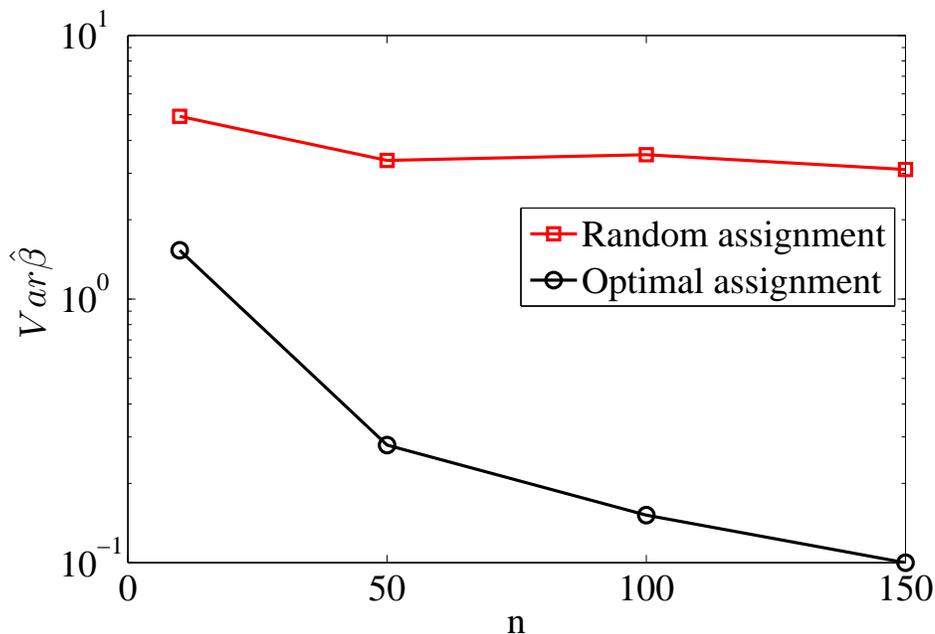}
		\caption{Variance of $\hat{\beta}$ as a function of $n$ when $p = n-1$, simulated from building data, with two different assignment strategies.}
		\label{building}
	\end{figure}

	The variance $\Var \hat{\beta}$ is a varying quantity in terms of the dimension of the covariates, i.e., $p-1$. Assume that the number of users are fixed, i.e., $n$ is fixed, and that $\rho$ is a constant. According to Theorem \ref{theorem1} and in Theorem \ref{theorem2}, with random assignment the variance decays with $n-p$ whereas with optimal assignment the worse case scenario with $p = n - 1$ yields a decay rate of $n$.

	In Fig. \ref{varying_d}, we show how $\Var \hat{\beta}$ varies with an increasing $p$ and a fixed $n = 50$. As can be seen from Fig. \ref{varying_d}, the increasing $p$ deteriorates the variance much severely from random assignment than that from the proposed optimal assignment. This again validates the efficiency of the proposed strategy in improving estimation accuracy, especially in a high dimensional setting where $p$ is comparable to $n$.

		\begin{figure}[!t]
			\centering
			\includegraphics[width = 0.7\columnwidth]{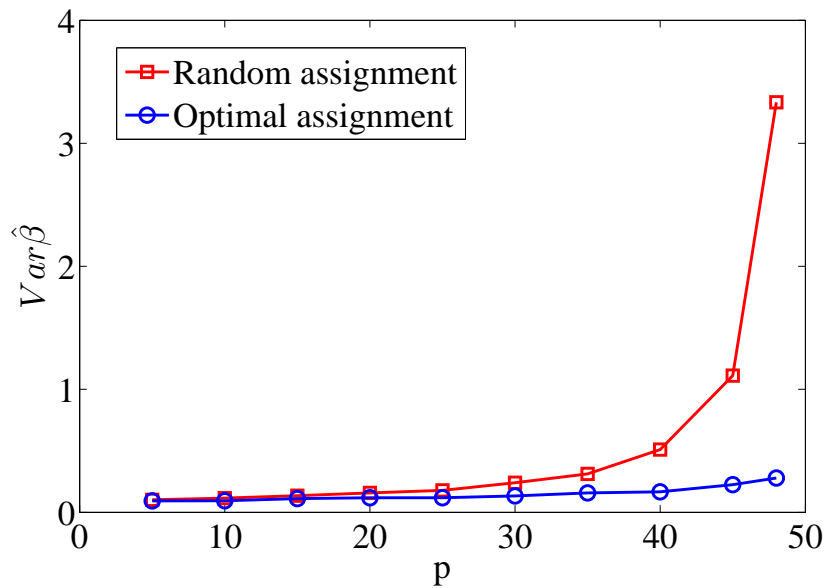}
			\caption{Variance of $\hat{\beta}$ as a function of $p$ when $n$ is fixed, simulated from building data, with two different assignment strategies.}
			\label{varying_d}
		\end{figure}

	\section{CONCLUSION}\label{conclusion}
	In this paper, we estimate the average treatment effect of demand response (DR) signals. We adopt an additive linear regression model and discuss two different strategies to assign DR signals to users under limited assignment budgets. The first strategy randomly picks $k$ users and sends DR signals to them. The second strategy optimally assigns DR signals to $k$ users by minimizing the variance to estimate treatment effect. We show that in a high dimensional setting, the second strategy achieves order optimal rates in variance reduction, whereas random assignment does not reduce variance even as the number of users grows. We formulate the general assignment as a combinatorial optimization problem and present a tractable SDP relaxation. We show that this relaxation obtains a solution that is within bounded gap of the original optimal solution. The simulation results validate this proposition with the synthetic data on both i.i.d. Gaussian covariates and uniform covariates. This work provides a framework for further research in applying causal inference in analyzing consumption data and DR interventions.

	% by increasing the rate of $\Var^{-1}\hat{\beta}$ from O($1$) to O($n$). We then relax the optimization problem in this strategy into SDP formulation and argue that the performance ratio is no worse than $\frac{k}{n}$ for a wide range of $\frac{k}{n}$'s. The simulation results validate this proposition with the simulation on i.i.d. Gaussian covariates and uniform covariates. The simulation also shows that the solution by SDP relaxation is very close to the result obtained by branch and bound, whereas both an intuitive greedy algorithm and the random assignment fail to recover an approximate solution. This work can provide a framework for further research in applying causal inference in analyzing consumption data and DR interventions.

	\bibliographystyle{IEEEtran}	% (uses file "plain.bst")
	\bibliography{ref}

% \newpage
% \onecolumn
% \doublespacing

	\section*{Appendix}
	\subsection{Proof of Lemma \ref{Gaussiannull}}
	\begin{proof} %[Proof of Lemma \ref{Gaussiannull}]
		This proof follows the intuition in \cite{nullspace}. In \cite{nullspace}, to prove the null space of a Gaussian random matrix, the authors use the fact that a standard multivariate Gaussian distribution is invariant to any orthogonal transform. The difference here is that the matrix $\bm{Z}$ contains an extra column of 1's as intercept.

		Denote the null space of $\bm{Z}_{n,n-1}^{\text{T}}$ by $\bm{y}$. It is one dimensional and satisfies:
		\begin{equation}\label{eq:null}
		\begin{bmatrix}
		\bm{Z}_{1:n-2}^{\text{T}}\\
		\bm{1}^{\text{T}}
		\end{bmatrix} \bm{y} = \bm{0}.
		\end{equation}

		Now we multiply $\bm{y}$ by an orthogonal matrix $\bm{U}_{n,n}^{\text{T}}$ on the left hand side and multiply $\bm{Z}_{n,n-1}^{\text{T}}$ by $\bm{U}_{n,n}$ on the right hand side. For simplicity let us write $\bm{U}_{n,n}$ as $\bm{U}$. Since $\bm{U}\bm{U}^{\text{T}} = {I}$ then the following holds:
		\begin{equation}\label{eq:transform}
		\begin{bmatrix}
		\bm{Z}_{1:n-2}^{\text{T}}\\
		\bm{1}^{\text{T}}
		\end{bmatrix} \bm{U}\bm{U}^{\text{T}}\bm{y} =
		\begin{bmatrix}
		\bm{Z}_{1:n-2}^{\text{T}}\bm{U}\\
		\bm{1}^{\text{T}}\bm{U}
		\end{bmatrix} \bm{U}^{\text{T}}\bm{y}
		= \bm{0}.
		\end{equation}

		Because $\bm{Z}_{1:n-2}^{\text{T}}$ has i.i.d. standard Gaussian entries and $\bm{U}$ is an orthogonal matrix, then each row of $\bm{Z}_{1:n-2}^{\text{T}}\bm{U}$ must follow a Gaussian distribution $N(\bm{0}\bm{U}, \bm{U}{I}\bm{U}^{\text{T}})$, which is still a standard Gaussian distribution. Thus $\bm{Z}_{1:n-2}^{\text{T}}\bm{U}$ has i.i.d. standard Gaussian entries. Denote $\tilde{\bm{Z}} = \bm{Z}_{1:n-2}^{\text{T}}\bm{U}$, then $\tilde{\bm{Z}} $ and $\bm{Z}_{1:n-2}$ are drawn from the same distribution and each of their entries follows i.i.d. standard Gaussian distribution.

		From \eqref{eq:transform}, we also require that the orthogonal matrix $\bm{U}$ satisfies $\bm{1}^{\text{T}}\bm{U} = \bm{1}$. If such orthogonal matrix exists (which is easy to find), then we can rewrite \eqref{eq:transform} as:
		\begin{equation}\label{eq:newnull}
		\begin{bmatrix}
		\tilde{\bm{Z}}^{\text{T}}\\
		\bm{1}^{\text{T}}
		\end{bmatrix} \bm{U}^{\text{T}}\bm{y}
		= \bm{0}.
		\end{equation}

		Comparing \eqref{eq:newnull} and \eqref{eq:null}, we see that $\bm{U}^{\text{T}}\bm{y}$ and $\bm{y}$ must be identically distributed. One distribution that satisfies this property is the standard normal distribution $N(\bm{0}, {I})$. It is easy to show that identity matrix is the only covariance matrix that satisfies such condition whereas zero mean is based on the fact that $\bm{1}^{\text{T}}\bm{y} = 0$ from \eqref{eq:null}. This observation concludes the final proof.
	\end{proof}

	\subsection{Proof of Lemma \ref{lemma2}}

	\begin{proof}
		The proof is based on the fact that if in expectation, at least $k$ $y_i$'s is greater than a constant, then the $k$th largest $y_i$ must be greater than this constant \cite{MaxGaussian}. Mathematically speaking, we want $\E \{ |i:y_i \geq C \sqrt{\log{\frac{n}{k}}}, \forall i|\} \geq k$, where $C$ is a constant. We write $\E \{ |\{i:y_i \geq C \sqrt{\log{\frac{n}{k}}}\}|\}$ in the following form:
		\begin{equation}\label{lowerbound1}
		\begin{aligned}
		\E \{ |\{i:y_i \geq C \sqrt{\log{\frac{n}{k}}}\}|\}  & = n \text{Pr} \{y_i \geq C \sqrt{\log{\frac{n}{k}}}\} \\
		& = nQ(C \sqrt{\log{\frac{n}{k}}}).\\
		\end{aligned}
		\end{equation}

		$Q$ function of Gaussian distribution does not have a close form expression, so we are interested in a tight lower bound for it in order to obtain a lower bound for \eqref{lowerbound1}. Many lower bounds are obtained in literature \cite{05470020, 1202.6483, erf-approx,j64}. In particular, \cite{erf-approx} provides a lower bound that is valid when the argument is small, and \cite{05470020} provides a lower bound that is the tightest when the argument becomes relatively large. The lower bounds (bound 1 through bound 4, from \cite{05470020, 1202.6483, erf-approx,j64} respectively) are presented in \eqref{bound1} through \eqref{bound4}.
			\begin{subequations}
				\begin{align}
				\label{bound1}
				Q(x) &  \geq \frac{1}{12}e^{-x^2}+\frac{1}{\sqrt{2\pi}(x+1)}e^{-\frac{x^2}{2}},\\
				\label{bound2}
				Q(x)&\geq  \frac{e^{\frac{1}{\pi(\kappa-1)+2}}}{2\kappa}\sqrt{\frac{1}{\pi}(\kappa-1)(\pi(\kappa-1)+2)}e^{-\frac{\kappa x^2}{2}}, \kappa \geq 1,  \\
				\label{bound3}
				Q(x) & \geq \frac{1-\sqrt{1-e^{-\frac{2x^2}{\pi}}}}{2},\\
				\label{bound4}
				Q(x) & \geq \sqrt{\frac{e(\beta-1)}{2\pi\beta^2}}e^{-\frac{\beta x^2}{2}}, \beta \geq 1.
				\end{align}
			\end{subequations}

		Note that bound 3 (in \eqref{bound3}) is only valid when $x$ is small. A comparison of these four bounds are shown in Fig.\ref{lowerbound}. It is shown in semilogarithmic plot where $y$-axis is in a logarithmic scale and $x$-axis is in a linear scale.

		Assume that $C \sqrt{\log{\frac{n}{k}}}$ is small, i.e., when $\frac{k}{n} = \frac{1}{2} - \epsilon$ which is slightly smaller than $\frac{1}{2}$ but greater than $\frac{1}{4}$. In this range the lower bound provided in \eqref{bound3} is the tightest. Use this lower bound we obtain the constant $C$ that is universally applicable for that \eqref{lowerbound1} holds when $\frac{n}{k}$ is small. After some simple calculation, we obtain the constant $C = \sqrt{\frac{\pi \log{(1-4\epsilon^2)}}{2\log{(\frac{1}{2}-\epsilon)}}}$.

		Although \eqref{bound1} provides the tightest bound when $C \sqrt{\log{\frac{n}{k}}}$ gets bigger, this lower bound includes two exponential terms which complicate the calculation. Actually when $\frac{k}{n}\leq \frac{1}{4}$, we can obtain a fairly good but conservative constant $C$ for the lower bound provided in \eqref{bound4} with only one exponential term. The constant $C$ in this case when $\frac{n}{k}$ is relatively large is calculated as $\sqrt{1-\frac{\log{8\pi}-1}{2\log{4}}} \approx 0.445$.

		\begin{figure}[!t]
			\centering
			\includegraphics[width = 0.8\columnwidth]{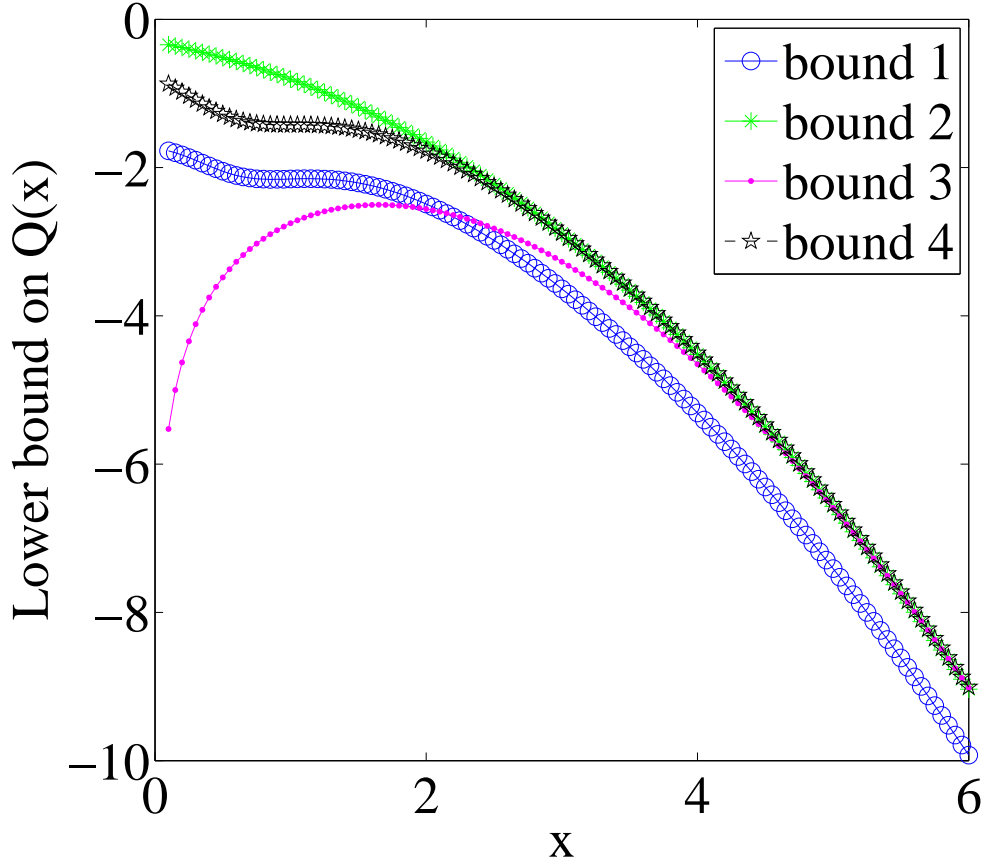}
			\caption{Semilogarithmic plot of the lower bounds discussed in \eqref{bound1} to \eqref{bound4}. The lower bounds are annotated in the graph as bound 1 to bound 4, respectively.}
			\label{lowerbound}
		\end{figure}

	\end{proof}

	\subsection{Proof of Lemma \ref{selfedge}}
	Before proving Lemma \ref{selfedge}, let us first introduce Lemma \ref{wii} that bounds the values of the diagonals of $P_{z^{\perp}}$.
	\begin{lemma}\label{wii}
		The diagonals of $P_{z^{\perp}}$ is non negative.
	\end{lemma}
	\begin{proof}[Proof of Lemma \ref{wii}]
		Recall that $\bm{P}_{\bm{z}^{\perp}}  = \bm{I} -  \bm{Z}(\bm{Z}^{\text{T}}\bm{Z})^{-1}\bm{Z}^{\text{T}}$, and $\bm{Z}(\bm{Z}^{\text{T}}\bm{Z})^{-1}\bm{Z}^{\text{T}}$ is a projection matrix, then we have the following:
%		\begin{equation}
%			\begin{aligned}
%			\bm{1}^{\text{T}}(\bm{Z}(\bm{Z}^{\text{T}}\bm{Z})^{-1}\bm{Z}^{\text{T}})^2 & = \bm{1}^{\text{T}}\bm{Z}(\bm{Z}^{\text{T}}\bm{Z})^{-1}\bm{Z}^{\text{T}}\\
%			& = \bm{1}^{\text{T}}.
%			\end{aligned}
%		\end{equation}

%Denote the diagonals of the $\bm{Z}(\bm{Z}^{\text{T}}\bm{Z})^{-1}\bm{Z}^{\text{T}}$ by $v_{ii}$, then $\sum_i v_{i1}^2 = \sum_i v_{i1}$. Generalize to all $i$, we have that $v_{ii} \leq 1$.

			\begin{equation}
			\begin{aligned}
			(\bm{Z}(\bm{Z}^{\text{T}}\bm{Z})^{-1}\bm{Z}^{\text{T}})^2 & = \bm{Z}(\bm{Z}^{\text{T}}\bm{Z})^{-1}\bm{Z}^{\text{T}}.\\
			\end{aligned}
			\end{equation}

Denote the diagonals of the $\bm{Z}(\bm{Z}^{\text{T}}\bm{Z})^{-1}\bm{Z}^{\text{T}}$ by $v_{ii}$, then we have $v_{ii} = \sum_j v_{ij}^2 = v_{ii}^2 + \sum_{j \neq i} v_{ij}^2$, which implies that $0 \leq v_{ii} \leq 1$.

		This suggests that the diagonals of $P_{z^{\perp}}$, denoted by $w_{ii}$, is non-negative.
	\end{proof}

	Now we proceed to prove Lemma \ref{selfedge}.

		\begin{proof} [Proof of Lemma \ref{selfedge}]

			To prove that Lemma \ref{selfedge} holds, we add up a constant to each $w_{ij}$ to make sure that there are no negative weights in the graph. This does not change the optimization problem since it only adds a constant to the objective function and the solution remains the same.

			If $|S| \leq k$, then we arbitrarily add vertices to $\tilde{S}$ until it contains exactly $k$ vertices. Since the weight on each edge is non negative and we keep adding more edges into the subgraph induced by $\tilde{S}$, $\xi$ is at least 1 in this case.

			Now suppose that $|S| > k$, in this case we need to eliminate vertices from $\tilde{S}$ until it only contains $k$ vertices. Assume that we want to eliminate vertex $i$ from $\tilde{S}$, then the total weights induced by $\tilde{S} \setminus\{ i\} $ are:
			\begin{equation}
			w(\tilde{S}) - (\sum_{j \in \tilde{S}, i \neq j} w_{ij} + w_{ii}).
			\end{equation}

			If each vertex is removed once, then:
			\begin{equation}
			\begin{aligned}
			\sum_{i \in \tilde{S}} w(\tilde{S}\setminus \{i\}) & = \sum_{i \in \tilde{S}} w(\tilde{S}) - (\sum_{j \in \tilde{S}, i \neq j} w_{ij} + w_{ii}) \\
			& = (|\tilde{S}|-2)w(\tilde{S})+\sum_{i \in \tilde{S}}w_{ii}.
			\end{aligned}
			\end{equation}

			The last equality is because each non-self edge is counted twice during the removal (once at the count of each vertex), but self-edge is only counted once.

			Similarly, suppose that $v$ is the node that is removed during the swapping procedure, then according to the swapping procedure in Algorithm \ref{algo}:
			\begin{equation}
			\sum_{j \in \tilde{S}} w_{vj} \leq \sum_{j \in \tilde{S}}w_{ij}, \forall i \neq v \in \tilde{S}.
			\end{equation}

			Then:
			\begin{equation}
			\begin{aligned}
			w(\tilde{S}\setminus \{v\}) & \geq \frac{1}{|\tilde{S}|} \sum_{i \in \tilde{S}}w(\tilde{S}\setminus \{i\}) \\
			& = \frac{1}{|\tilde{S}|}((|\tilde{S}|-2)w(\tilde{S})+\sum_{i \in \tilde{S}}w_{ii}) \\
			& = \frac{|\tilde{S}|-2}{|\tilde{S}|}w(\tilde{S}) + \frac{\sum_{i \in \tilde{S}}w_{ii} }{|\tilde{S}|} \\
			& \geq \frac{|\tilde{S}|-2}{|\tilde{S}|}w(\tilde{S}).
			\end{aligned}
			\end{equation}

			The last inequality follows because $w_{ii}$'s are non-negative, which is proved in Lemma \ref{wii}.

			Finally by induction, we obtain the eventual $\tilde{S}$ containing $k$ vertices satisfying:
			\begin{equation}
			w(\tilde{S}) \geq \frac{k(k-1)}{|S|(|S|-1)}w(S),
			\end{equation}
			which concludes the final proof.
		\end{proof}

		\subsection{Whitening of covariate data}
		Suppose that the covariates generated from the simulation tool in \cite{sim} has a covariance matrix $\Sigma$. We can decompose $\Sigma = \bm{V}\bm{D}\bm{V}^{\top}$， where $\bm{D}$ is a diagonal matrix with diagonal elements as the eigenvalues of $\Sigma$, and $\bm{V}$ is an orthogonal matrix containing the corresponding eigenvectors. Suppose that each observation of the covariates is denoted by $\bm{z}$, then the whitened observation $\bar{\bm{z}} = \bm{D}^{-\frac{1}{2}}\bm{V}^{\top}\bm{z} \in \mathbb{R}^{p-1}$, where $p-1$ is the dimension of the covariates. For the original covariate matrix $\bm{Z}$, it is whitened as $\bar{\bm{Z}} = \bm{Z}\bm{V}\bm{D}^{-\frac{1}{2}}$. The whitened covariates are then fed into the optimization problem to find the optimal assignment strategy.

% \newpage
% This is a resubmission of the manuscript 17-0116 because of a formatting issue on the first submission. No reviews have been received.

\end{document}